\documentclass[runningheads]{llncs}
\usepackage[T1]{fontenc}
\usepackage{graphicx}
\usepackage{tikz}

\usepackage{amsmath}
\usepackage{amssymb}

\usepackage{verbatim}
\usepackage{url}
\usepackage{tikz}

\newcommand{\ta}{\alpha}
\newcommand{\taa}{\beta}

\spnewtheorem{observation}{Observation}{\bfseries}{\rmfamily}

\begin{document}
\title{Attributed Tree Transducers for Partial Functions} 
%
%

\author{Sebastian Maneth \and
  Martin Vu}
\authorrunning{S. Maneth and
  M. Vu}
%
\institute{Universit\"at Bremen, Germany\\
  \email{\{maneth,martin.vu\}@uni-bremen.de}
}

\maketitle

\begin{abstract}
Attributed tree transducers (atts) have been equipped with regular look-around
(i.e., a preprocessing via an attributed relabeling) in order to obtain
a more robust class of translations. Here we give further evidence of
this robustness: we show that if the class of translations realized by nondeterministic
atts with regular look-around is restricted to partial functions, then we obtain
exactly the class of translations realized by deterministic atts with
regular look-around.
\end{abstract}
\section{Introduction}
Attributed tree transducers (atts)~\cite{DBLP:journals/actaC/Fulop81}
are a well known formalism for defining tree translations.
They are attractive, because they are strictly more expressive than
top-down tree transducers and they closely model the behavior of
attribute grammars~\cite{DBLP:journals/mst/Knuth68,Knuth71}.
Since attribute grammars are deterministic devices, atts also have typically been
studied in their deterministic version. 

But atts also have some deficiencies. For instance, they do not generalize
the (deterministic) bottom-up tree translations~\cite{DBLP:journals/iandc/FulopV95}.
One possibility to remedy this deficiency is to equip atts with 
\emph{regular look-around}~\cite{DBLP:journals/jcss/BloemE00}.
	The resulting class of translations is more robust in the sense that
	(\emph{i})~it does generalized deterministic bottom-up tree translations~\cite{DBLP:journals/mst/Engelfriet77},
	(\emph{ii})~it is equivalent to tree-to-dag-to-tree translations that are definable in 
	MSO logic~\cite{DBLP:journals/jcss/BloemE00}, and 
	(\emph{iii})~it is characterized by natural restrictions of 
	macro tree transducers~\cite{DBLP:journals/tcs/HashimotoM23}.

In this paper we present another advantage of adding regular look-around:
we show that every nondeterministic att (with or without regular look-around) that is
\emph{functional} can be realized by a \emph{deterministic} att
with regular look-around (and such a transducer can be constructed).
In general it is a desirable and convenient property of a nondeterministic translation device
that its restriction to functional translations coincides precisely with the translations
realized by the corresponding deterministic version of that device.
Let us consider some classical examples of translation devices for which this property
holds:
\begin{itemize}
	\item
	two-way strings transducers (``2GSM'')~\cite[Theorem~22]{Engelfriet1999MSODS}\cite[Theorem~3]{DBLP:conf/lata/Souza13},
	\item
	top-down tree transducers with look-ahead~\cite[Theorem~1]{DBLP:journals/ipl/Engelfriet78}, and
	\item
	macro tree transducers~\cite[Corollary~36]{DBLP:journals/acta/EngelfrietIM21}.
\end{itemize}
In contrast, these are examples for which the property does \emph{not} hold:
\begin{itemize}
	\item
	one-way string transducers (``GSM''),
	\item  
	top-down and bottom-up tree transducers, and
	\item
	attributed tree transducers.
\end{itemize}

To see that functional  top-down tree transducers and functional GSMs are strictly more 
expressive than their deterministic counterparts, consider the following translation:
Let monadic input trees of the form $a^n(e)$ be translated to $a^n(e)$ while input trees of the form
$a^n(f)$ are translated to the single leaf $f$.
A deterministic top-down tree transducer (or a  deterministic GSM) cannot realize this translation,
because it has to decide whether or not to output $a$-nodes, before it sees the label
of the input leaf (viz. the right end of the string).

Given a relation 
$R$, a \emph{uniformizer of $R$} is a function
$f_R\subseteq R$ such that $f_R$ and $R$ have the same domain.
Engelfriet~\cite{DBLP:journals/ipl/Engelfriet78} showed that for any nondeterministic top-down tree translation $R$,
a uniformizer of $R$ can be constructed as a deterministic top-down tree transducer
with look-ahead ($dt^R$) . Since the latter type of transducer is closed under composition~\cite{DBLP:journals/mst/Engelfriet77},
it follows that for any composition of top-down (and bottom-up) tree transductions
that is functional, an equivalent $dt^R$ can be constructed.
Whenever several rules of the given transducer are applicable, the uniformizer chooses
the first one of them (in some order); the regular look-ahead is used to determine
which rules are applicable.
For an overview over different uniformization results in automata theory see~\cite{paper6}.

Now 
consider a functional attributed tree transducer (att). Its input trees are over the
binary labels $f$ and $g$ and the leaf label $e$ and its output trees are over the unary label $d$
and the leaf label $e$. 
The att outputs
$d^n(e)$ if there exists a first $g$-node in reverse pre-order and $n$ is the size of the subree rooted at that node 
(no output is produced, if no
such $g$-node exists).
This translation cannot be realized by any deterministic att,
even if it uses regular look-ahead.
However, there is a deterministic att with \emph{regular look-around} that realizes the translation.

Subsequently, we show that  any  (possibly circular) nondeterministic att $A$  effectively has a uniformizer that is realized
by a deterministic att with look-around.
To construct such a deterministic att $D$ with regular look-around, we first construct
a nondeterministic (and in general non-functional) top-down relabeling $T$ and
a deterministic att $D'$ which operate as follows: 
Informally, the top-down relabeling $T$ specifies which rules the deterministic att $D'$ should apply.
More precisely,
given an input tree $s$,
$T$ annotates all nodes of $s$ by rules of $A$. 
At every node, the deterministic att $D'$ applies exactly those rules with which the nodes are annotated.  
In particular, $T$ annotates the nodes of $s$ such that $D'$ simulates a \emph{uniform} translation of $A$,
that is, a translation in which if multiple instance of an attribute $\ta$ access the same node $v$,
then all these instances of $\ta$ need to apply the same rule. 
With the result of~\cite{DBLP:journals/ipl/Engelfriet78}, 
a look-around $U$ can be constructed that realizes a uniformizer of the translation of $T$. 
With the deterministic att $D'$, this look-around $U$ yields $D$.  

The existence of a uniformizer realized by a deterministic att with regular look-around
for any att implies that a composition of $n$ atts (with or without look-around)
that is functional can be simulated by a composition of $n$ deterministic atts with look-around.
This also implies that for any functional att there is effectively an equivalent deterministic att with look-around.

An abridged version of this paper has been accepted for publication at the Conference on Implementation and Application of Automata
(CIAA 2024).

\section{Preliminaries}\label{preliminaries}
Denote by $\mathbb{N}$ the set of natural numbers.
For $k\in\mathbb{N}$, we denote by $[k]$ the set $\{1,\dots,k\}$.

A set $\Sigma$ is \emph{ranked} if  
each symbol of $\Sigma$ is associated with a \emph{rank}, that is, a non-negative integer.
We write $\sigma^{k}$ to denote that the symbol $\sigma$ has
rank~$k$.
By $\Sigma_k$ we denote the set of  all symbols of $\Sigma$ which have rank $k$.
For $k'\neq k$, we define that $\Sigma_{k'}$ and $\Sigma_k$ are disjoint.
If the set $\Sigma$ is finite then we call $\Sigma$ a \emph{ranked alphabet}. 

The set $T_\Sigma$ of 
\textit{trees over $\Sigma$} is defined as the smallest set of strings such that
if $\sigma\in \Sigma_k$, $k\geq 0$, and $t_1,\dots,t_k \in T_\Sigma$ then
$\sigma(t_1,\dots, t_k)$ is in $T_\Sigma$. 
For $k=0$, we simply write $\sigma$
instead of $\sigma()$.

The nodes of a tree $t\in T_\Sigma$ are referred to by strings over~$\mathbb{N}$.
In particular, for $t=\sigma (t_1,\dots,t_k)$, we define $V(t)$, the set of nodes of $t$, 
as $V(t)=\{\epsilon\}\cup \{iu\mid i\in [k] \text{ and } u\in V(t_i)\}$,
where $\epsilon$ is the \emph{empty string}.
For better readability, we add dots between numbers,
e.g. for the tree $t=f(a,f(a,b))$ we have $V(t)=\{\epsilon,1,2,2.1,2.2\}$.
Let $v,v'\in V(t)$.
Then we call $v$ a (proper) ancestor of a node $v'$ if $v$ is a (proper) prefix of~$v'$.

For a node  $v\in V(t)$, $t[v]$ denotes the label of $v$, 
$t/v$ is the subtree of $t$ rooted at $v$, and
$t[v \leftarrow t']$ is  obtained from $t$ by replacing $t/v$ by~$t'$.
For instance, for $t=f(a,f(a,b))$,
we have $t[1]=a$, $t/2= f(a,b)$ and $t[1\leftarrow b]=f(b,f(a,b))$.
The \emph{size} of a tree $t$ is given by $\text{size} (t)=|V(t)|$.

For a set $\Lambda$ disjoint with $\Sigma$, we define $T_\Sigma [\Lambda]$ as $T_{\Sigma'}$ where $\Sigma'_0 =\Sigma_0\cup \Lambda$ and $\Sigma_k'=\Sigma_k$ for $k>0$.

Let $R\subseteq A\times B$ be a relation.
We call $R$ a \emph{function} if $(a,b), (a,b') \in R$ implies $b=b'$.
We define the  \emph{domain of $R$} by $\text{dom} (R)=\{a\in A\mid \exists b\in B:\ (a,b)\in R\}$.
Analogously,  the \emph{range of $R$} is $\text{range} (R)=\{b\in B\mid \exists a\in A:\ (a,b)\in R\}$.
A function $F\subseteq R$ is called a \emph{uniformizer} of $R$ if  $\text{dom} (F) = \text{dom} (R)$.
Let $R'\subseteq B\times C$. The composition of $R$ and $R'$ is $R\circ R'= \{(a,c) \mid (a,b)\in R, (b,c)\in R'\}$.

\section{Attributed Tree Transducers}
In the following, we define attributed tree transducers. For an in-depth introduction to attributed tree transducers,
we refer to~\cite{DBLP:series/eatcs/FulopV98}.

	A \emph{(partial nondeterministic) attributed tree transducer} (or $att$ for short) is a tuple $A=(S,I,\Sigma,\Delta,a_0,R)$,
	 where
	 \begin{itemize}
	 	\item  $S$ and $I$ are disjoint finite sets of
	 	\emph{synthesized attributes} and \emph{inherited attributes}, respectively,
	 	\item $\Sigma$ and $\Delta$ are ranked alphabets of \emph{input} and \emph{output symbols}, respectively,
	 	\item  $a_0\in S$ is  the \emph{initial attribute} and
	 	\item  $R=(R_\sigma\mid \sigma \in \Sigma\cup \{\#\})$ is a collection of finite sets of rules.
	 \end{itemize}
	We implicitly assume $atts$ to include a unique symbol $\#\notin \Sigma$ of rank $1$, the so-called
	\emph{root marker}, that only occurs at the root of trees.

	In the following, we define the rules of an $att$.
	Let $\sigma\in \Sigma$ be of rank $k\geq 0$. 
	Furthermore, let $\pi$ be a variable for nodes. Then the set $R_\sigma$ contains
	\begin{itemize}
		\item 	arbitrarily many rules of the form $a(\pi)\rightarrow \xi$ for every $a\in S$ and
		\item  arbitrarily many rules of the form $b(\pi i)\rightarrow \xi'$ for every $b\in I$ and $i\in [k]$,
	\end{itemize}
 where
	$\xi,\xi'\in T_\Delta [\{a'(\pi i) \mid a'\in S, i\in [k]\} \cup \{b'(\pi) \mid b'\in I\} ]$.
	We define the set $R_\#$ analogously with the restriction that $R_\#$ contains \emph{no}
	rules with synthesized attributes  on the left-hand side and  inherited attributes on the right-hand side.
	Replacing  `arbitrarily many rules' by `at most one rule' in the definition of the
	rule sets of $R$, we obtain the notion of a \emph{(partial) deterministic $att$} (or $datt$ for short).
	For the $att$ $A$ and the attribute $a\in S$, we denote by $\text{RHS}_A (\sigma, a(\pi))$
	the set of all right-hand sides of rules in $R_\sigma$ that are of the form
	$a(\pi)\rightarrow \xi$.
	For $b\in I$, the sets $\text{RHS}_A (\sigma, b(\pi i))$ with $i\in [k]$ and $\text{RHS}_A (\#, b(\pi 1))$ are defined analogously.\\
		
\textbf{Attributed Tree Translation.}\quad
We now define the semantics of $A$. 
Denote by $T_{\Sigma^\#}$ the set $\{\# (s) \mid s\in T_\Sigma\}$.
For a given tree $s\in T_{\Sigma^\#}$, we define
\[
\text{SI}(s)= \{\ta (v) \mid \ta\in S\cup I, v\in V(s)\}.
\]
Furthermore, we define that
 for the node variable $\pi$,
$\pi 0 = \pi$  and that for a node $v$,
$v.0=v$.
Let $t,t' \in T_\Delta [\text{SI}(s)]$. 
We write $t\Rightarrow_{A,s} t'$ if 
$t'$ is obtained from $t$
by substituting  a leaf of $t$ labeled by $\gamma (v.i)$ with $i=0$ if $\gamma\in S$ and $i>0$ if $\gamma\in I$
by 
$\xi [\pi\leftarrow v]$,
where $\xi\in \text{RHS}_A (s[v], \gamma (\pi i))$ and 
$[\pi\leftarrow v]$ denotes the substitution that replaces all occurrences
of $\pi$ by the node $v$. For instance, for $\xi_1=f(b(\pi))$ and $\xi_2=f(a(\pi 2))$ where $f$ is a symbol of
rank~$1$, $a\in S$
and $b\in I$, we have $\xi_1[\pi\leftarrow v]=f(b(v))$ and $\xi_2[\pi\leftarrow v]=f(a(v.2))$.
Denote by 
$\Rightarrow_{A,s}^+$ and
$\Rightarrow_{A,s}^*$ the transitive closure and the reflexive-transitive closure of $\Rightarrow_{A,s}$, respectively.

The 
\emph{translation realized by $A$}, denoted by $\tau_A$, is the relation
\[
\{(s,t) \in T_\Sigma \times T_\Delta\mid  a_0(1)\Rightarrow_{A,s^\#}^* t\},
\]
where subsequently $s^\#$ denotes the tree $\# (s)$.
If $\tau_A$ is a  partial function then we say that $A$ is a \emph{functional} $att$.
We define $\text{dom} (A)=\text{dom} (\tau_A)$ and $\text{range} (A)= \text{range} (\tau_A)$.

The reader may wonder why the definition of $A$ and its translations involves the root marker.
Informally, the root marker is a technical requirement without which many translations
are not possible.
To see which role the root marker plays in a translation, consider the following translation
which cannot be realized without the root marker.

\begin{example}\label{att example}
	Consider the $att$ $A=(S,I,\Sigma,\Delta,a,R)$ where 
	$\Sigma= \{f^2, e^0\}$ and $\Delta=\{d^1, e^0\}$.
	We define $S=\{ a\}$ and $I=\{b\}$.
	For the symbol $f$, we define 
	\[
	R_f=\{   \underbrace{a(\pi)  \rightarrow   d (a (\pi 2))}_{\rho_1} ,\  \underbrace{b(\pi 2) \rightarrow  a (\pi 1)}_{\rho_2},\  
	 \underbrace{b(\pi 1) \rightarrow b (\pi)}_{\rho_3} \
	 \}.
	\]
	Finally, the  rule set for the symbol $e$ and the root marker  are  given by 
	\[
	R_e=\{ \underbrace{a(\pi) \rightarrow d (b(\pi))}_{\rho_4}\}
	\quad\quad\text{and}\quad\quad
	R_{\#}=\{ \underbrace{b(\pi 1) \rightarrow e}_{\rho_5}\},
	\]
	respectively.
	The tree transformation realized by $A$ is defined as follows:
	On input $s\in T_\Sigma$, where  $s$ is of size $n$, $A$ outputs the tree $d^{n} (e)$.
	We denote by $d^n (e)$ the tree over $\Delta$ with exactly $n$ occurrences of $d$, e.g.,
	$d^4 (e)$ denotes the tree $d(d(d( d (e))))$.
	For instance,
	for $s=f(e,e) $, the $att$ $A$ outputs the tree $d^3 (e)$.
	The corresponding translation is shown in Figure~\ref{fig 1}.
\end{example}

\begin{figure}
	\centering
    \begin{minipage}{0.4\textwidth}
	\begin{tikzpicture}
\draw (0,0) node[circle,inner sep=0.5pt,draw] {$\#$};
\draw (0,-0.3)--(0,-0.7);
\draw (0,-1) node[circle,inner sep=1pt,draw] {$f$};
\draw (-0.2,-1.3)--(-1.5,-1.7);
\draw (-1.5,-2) node[circle,inner sep=2pt,draw] {$e$};
\draw (0.2,-1.3)--(1.5,-1.7);
\draw (1.5,-2) node[circle,inner sep=2pt,draw] {$e$};

	\filldraw [black] (0.7,-1) circle (1.2pt);
	\draw (1,-1) node {$a$};
	\filldraw [black] (-0.7,-1) circle (1.2pt);
	\draw (-1,-1) node {$b$};
	\filldraw [black] (-0.8,-2) circle (1.2pt);
	\filldraw [black] (-2.2,-2) circle (1.2pt);
	\filldraw [black] (0.8,-2) circle (1.2pt);
	\filldraw [black] (2.2,-2) circle (1.2pt);
	
		\draw (-0.7,-0.2) node {$e$};
		\draw[->] (-0.7,-0.5)--(-0.7,-0.8);
		\draw [->](-0.7,-1.2).. controls(-2.2,-1.5)..(-2.2,-1.8);
		\draw[->] (-2.2,-2.2).. controls(-1.5,-2.5)..(-0.8,-2.2);
		\filldraw[white] (-1.5,-2.5) circle (5pt);
		\draw (-1.5,-2.5) node {$d$};
		
		\draw [->](-0.8,-1.9).. controls(0,-1.6)..(0.8,-1.9);
		
		\draw [<-](2.2,-2.2).. controls(1.5,-2.5)..(0.8,-2.2);
		\filldraw[white] (1.5,-2.5) circle (5pt);
		\draw (1.5,-2.5) node {$d$};
		
		\draw [<-](0.8,-1.1).. controls(2.2,-1.5)..(2.2,-1.8);
		\filldraw[white] (1.7,-1.4) circle (5pt);
		\draw (1.7,-1.4) node {$d$};
	\end{tikzpicture}
    \end{minipage}
\quad\quad
    \begin{minipage}{0.5\textwidth}
	\begin{tikzpicture}	
\draw (0,0) node {$a(1)$};
\draw (1,0) node {$\Rightarrow$};
\draw (1,-0.25) node {$\rho_1$};

\draw (2.5,0) node {$d (a (1.2))$};
\draw (4,0) node {$\Rightarrow$};
\draw (4,-0.25) node {$\rho_4$};

\draw (5.5,0) node {$d^2 (b(1.2))$};

\draw (1,-1) node {$\Rightarrow$};
\draw (1,-1.25) node {$\rho_2$};

\draw (2.5,-1) node {$d^2 (a (1.1))$};
\draw (4,-1) node {$\Rightarrow$};
\draw (4,-1.25) node {$\rho_4$};

\draw (5.5,-1) node {$d^3 (b(1.1))$};

\draw (1,-2) node {$\Rightarrow$};
\draw (1,-2.25) node {$\rho_3$};

\draw (2.5,-2) node {$d^3 (b(1))$};
\draw (4,-2) node {$\Rightarrow$};
\draw (4,-2.25) node {$\rho_5$};

\draw (5.5,-2) node {$d^3 (e)$};
	\end{tikzpicture}
\end{minipage}
	\caption{The translation from  $f(e,e)$ to $d^3 (e)$
		defined by $A$ is pictured on the left. The corresponding transitions are displayed on the right.
		Each $\Rightarrow$ is annotated with the rule used in the corresponding transition step.
	}
	\label{fig 1}
\end{figure}
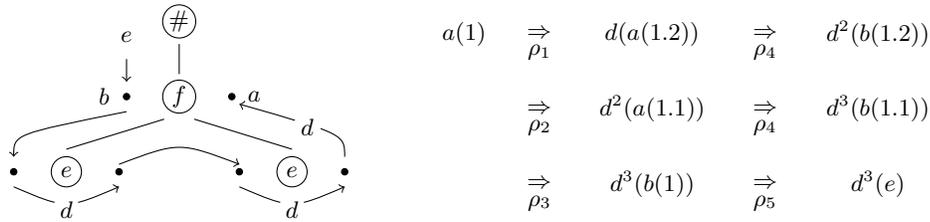

Note that by definition $atts$ are allowed to be \emph{circular}. 
We say that an $att$ $A$ is \emph{circular} if  $s\in T_{\Sigma}$, $\ta (v) \in \text{SI} (s^\#)$ and $t\in T_\Delta [\text{SI} (s^\#)]$  exists such that
$\ta (v) \Rightarrow_{A, s^\#}^+ t$ and $\ta (v)$ occurs in~$t$.\\

\textbf{Look-Ahead.}\quad To define \emph{attributed tree transducer with look-ahead}, we first define
\emph{bottom-up relabelings}.
A \emph{bottom-up relabeling} $B$ is a tuple $(P,\Sigma,\Sigma',F,R)$
where $P$ is the set of states, $\Sigma$ and $\Sigma'$ are finite ranked alphabets
and $F\subseteq P$ is the set of final states. 
For $\sigma\in \Sigma$ and $p_1,\dots,p_k\in P$,
the set $R$ contains 
at most one rule of the form
\[
\sigma (p_1 (x_1),\dots, p_k (x_k))\rightarrow p (\sigma' (x_1,\dots,x_k))
\]
where $p\in P$ and $\sigma'\in \Sigma'$.
The rules of $B$ induce a derivation relation $\Rightarrow_B^*$ defined inductively as follows: 
\begin{itemize}
	\item Let $\sigma\in \Sigma_0$ and $\sigma \rightarrow p (\sigma')$ be a rule in $R$.
	Then $\sigma \Rightarrow_B^*  p (\sigma')$.
	\item Let $s=\sigma (s_1,\dots, s_k)$ with $\sigma\in \Sigma_k$, $k>0$,
	and $s_1,\dots, s_k \in T_\Sigma$. 
	For $i\in [k]$, let $s_i \Rightarrow_B^* p_i (s_i')$. 
	Furthermore, let $\sigma (p_1 (x_1),\dots, p_k (x_k))\rightarrow p (\sigma' (x_1,\dots,x_k))$ be a rule in
	$R$. Then $s\Rightarrow_B^* p (\sigma' (s_1',\dots, s_k') )$.
\end{itemize}
For $s\in T_\Sigma$ and $p\in P$,
we write $s\in \text{dom}_B (p)$ if $s\Rightarrow_B^* p (s')$ for some  $s'\in T_{\Sigma'}$.
The translation realized by $B$ is 
\[
\tau_B= \{(s,s') \in T_\Sigma\times T_{\Sigma'} \mid s\Rightarrow_B^* p(s') \text{ with } p\in F\}.
\]

We define an \emph{attributed tree transducer with look-ahead (or $att^R$)}
as a pair
$\hat{A}=(B,A')$ where $B$ is a
bottom-up relabeling  and $A'$ is an $att$. 
Recall that translations are relations.
The translation realized by $\hat{A}$ is the composition of the translations realized by $B$ and $A'$, that is,
\[
\tau_{\hat{A}}=
\{ (s,t) \in T_\Sigma \times T_\Delta \mid (s,s')\in \tau_B (s) \text{ and } (s', t) \in \tau_{A'}  \}.
\]
Informally, $B$ preprocesses input trees for $A'$ with each node $v$, that is relabeled by $B$, providing $A'$ information about
the subtree rooted $v$. We say that $\hat{A}$ is a deterministic $att^R$ ($datt^R$) if $A$ is a $datt$.
The domain and the range of $\hat{A}$ are defined in the obvious way.\\

\textbf{Look-Around.}\quad
\emph{Look-around} is in concept similar to look-ahead; it is also a relabeling device that provides additional information to an $att$.
However, it is  far more expressive than a look-ahead\footnote{The concept of look-around is introduced in~\cite{DBLP:journals/jcss/BloemE00} where it is called look-ahead.}.

To define look-around, we first define  \emph{top-down relabelings}.
A top-down relabeling is a deterministic $att$ without inherited attributes
$T=(S,\emptyset,\Sigma ,\Sigma',a_0,R)$ where
for all  $\sigma\in \Sigma_k$, $k\geq 0$, the set $R_\sigma$
only contains rules of the form  $q (\pi)\rightarrow \sigma'(q_1(\pi 1),\dots,q_k (\pi k))$
with $\sigma'\in \Sigma'_k$.
Note that for a top-down relabeling, we will henceforth write
 $q(\sigma (x_1,\dots,x_k))\rightarrow \sigma'(q_1(x_1),\dots,q_k (x_k))$ instead
 $q (\pi)\rightarrow \sigma'(q_1(\pi 1),\dots,q_k (\pi k))\in R_\sigma$.
 We will also call the attributes of a top-down relabeling states.
 Since a top-down relabeling is an $att$, a \emph{top-down relabeling with look-ahead} is defined in the obvious way.

An \emph{attributed tree transducer with look-around} (or $att^U$)
is a tuple $\breve{A}=(U,A')$
where $A'$ is an $att$ and $U$ is a top-down relabeling with look-ahead.
The translation realized by $\breve{A}$
as well as the domain and the range of $\breve{A}$
 are defined analogously as for $att^R$.
This means that an $att^U$ relabels its input tree in two
phases: First the input tree is relabeled in a bottom-up phase. Then 
the resulting tree is relabeled again in a top-down phase before it is processed by $A'$.
We say that $\breve{A}$ is a deterministic $att^U$ ($datt^U$) if $A$ is a $datt$.
The following results hold.

\begin{proposition}\label{assumption}\cite[Lemma~27]{mypaper}
	Let $A$ be an  $att$. Then
	an equivalent $att$ $A'=(S,I,\Sigma,\Delta,a_0,R)$  can be constructed such that
	$R_\#$  contains no distinct rules with the same left-hand side.
\end{proposition}


\begin{proposition}\label{domain}\cite[Corollary~14]{DBLP:journals/acta/EngelfrietIM21}
	The domain of an $att^U$ is effectively regular.
\end{proposition}
\section{Functional Compositions of $\mathbf{Atts^U}$ are Determinizable}
In the remainder of this paper, we prove the following statement:
For any composition of $n$ $atts^U$ that is functional an equivalent composition of $n$
$datts^U$ can be constructed.
The overall idea is similar to the one in~\cite{DBLP:journals/ipl/Engelfriet78} for top-down tree transducers.
More precisely, let $C$ be the composition of the $atts^U$ $A_1,\dots,A_n$.
Subsequently, we show that (a) for $i\in [n]$, a $datt^U$ $D_i$ can be constructed such that $D_i$
realizes a uniformizer of $A_i$ and (b) the composition of $D_1,\dots,D_n$ is equivalent to $C$.
Obviously our statement implies that for any functional $att^U$ there is an equivalent $datt^U$.

Before we prove our statement, we remark that
in the absence of look-around our statement does not hold, i.e.,
a functional composition of $n$ $atts$ without look-around cannot necessarily be
simulated by a composition of $n$ $datts$ without-look-around.
In particular, we show that
there are functional $atts$ which cannot be simulated by
 $datts$ without look-around.

\subsection{On the Necessity of Look-Around}
Subsequently, we show that there are tree translations  realizable by functional
$atts$  that cannot be realized by any
$datt$ even if the $datt$ uses look-ahead. 
Consider the following example.

\begin{example}\label{example 2}
	Let $\Sigma=\{f^2,g^2,e^0\}$ and let $\Delta=\{d^1, e^0\}$.
	Consider the following tree translation.
	Let $s\in T_\Sigma$ and
	let $v$ be the first node in reverse pre-order that is labeled by $g$ in~$s$.
	Note that this means that no proper ancestor of $v$ is labeled by $g$ in~$s$.
	If such a node $v$ does not exist, i.e., if no input node is labeled by $g$, then no output is produced.
	Otherwise, we 
	translate $s$ into $d^m (e)$, where $m$ denotes the \emph{size} of $s/v$.
	Recall that the size of$s/v$ is the number of its nodes and
	that $d^m (e)$ denotes the tree over $\Delta$ with exactly $m$ occurrences of $d$, e.g.,
	$d^3 (e)$ denotes the tree $d(d(d(e)))$.
	Hence,  the tree $f(g(f (e,e),e), f (e, g(e, g(e,e))))$ is for instance translated into  $d^5 (e)$.
	Clearly, this tree translation is functional and can be realized by the following functional $att$.
	
	Let $A=(S,I,\Sigma, \Delta, a, R)$ where $S=\{a, a_g\}$ and $I=\{b,b_g, b_g'\}$.
	For the root marker and the symbol $e$, we define
	\[
	R_{\#}=\{b_g' (\pi 1) \rightarrow e \} \quad\text{and}\quad R_e=\{ a (\pi) \rightarrow b(\pi),\  a_g (\pi)  \rightarrow d (b_g(\pi) \},
	\]
	respectively.
	For the symbol $f$, we define
	\[
	\begin{array}{c lcl c lcl c lcl}
	R_f=\{ &	a (\pi)  & \rightarrow &a (\pi 2) & \quad &	b(\pi 2) & \rightarrow &  a (\pi 1) & \quad & b(\pi 1) & \rightarrow & b (\pi) \\
	 &	a_g (\pi)  & \rightarrow & d (a_g (\pi 2)) & \quad &	b_g(\pi 2) & \rightarrow &  a_g (\pi 1) & \quad & b_g(\pi 1) & \rightarrow & b_g (\pi) \\
	&	b_g'(\pi 2) & \rightarrow &  b_g' (\pi) & \quad & b_g'(\pi 1) & \rightarrow & b_g' (\pi)\ \}.
	\end{array}
	\]
	Finally, the rules for the symbol $g$ are given by
	\[
\begin{array}{ c lcl c lcl c lcl}
	R_g=\{ &	a (\pi)  & \rightarrow & d (a_g (\pi 2) ) & \quad & a_g (\pi)  & \rightarrow & d (a_g (\pi 2) ) & \quad & b_g (\pi 2) & \rightarrow &  a_g (\pi 1) \\	
 &	b_g(\pi 1) & \rightarrow & b_g (\pi) & \quad & 	b_g(\pi 1) & \rightarrow & b_g' (\pi) \ \}.\\
\end{array}
\]	
	In short, the attributes $a$ and $b$ are used for traversing the input tree in reverse pre-order until a node labeled by $g$
	is found.
	In the case that such a node is found, the
	attributes with subscript $g$ produce the output tree.
	Consider in particular, the rules $b_g (\pi 1)\rightarrow b_g (\pi)$ and $b_g (\pi 1)\rightarrow b_g' (\pi)$ in $R_g$.
	With these rules,  $A$ guesses whether the node labeled by $g$ that it is currently processing has a
	proper ancestor labeled by $g$ or not.
	In particular, applying the  rule $b_g (\pi 1)\rightarrow b_g (\pi)$ means that $A$ guesses that such
	a proper ancestor exists, while applying $b_g (\pi 1)\rightarrow b_g' (\pi)$ means that $A$ guesses the opposite.
\end{example}

Though the tree translation in Example~\ref{example 2} can by realized by a functional $att$,
no $datt$ $D$ can realize it. In the following, we will explain (without a formal proof) why this is the case.
The key point 
is that 
$D$ is unable to determine for an input tree $s$ and a given node $v$
whether or not a proper ancestor $v'$ of $v$ exists that is also labeled by $g$
\emph{and} in the case that such a $v'$ exists, continue to output symbols afterwards.
To specify this statement consider the following.
Obviously,  $D$ visits every node for which it has to produce a symbol $d$.
W.l.o.g. it can be assumed that $D$ visits all such nodes in reverse pre-order.
Assume that $D$ has already produced an output symbol $d$ for every descendant of the node $v$
and that an attribute of $D$ is currently processing $v$.
Now $D$ must determine whether or not any more symbols $d$ need to be produced.
In other words, $D$ must determine whether or not an ancestor of $v$ is labeled by $g$ or not.
To do so $D$ 
traverses $s$ from $v$ `upwards' using rules of the form $\taa(\pi j) \rightarrow \taa' (\pi)$
where $\taa$ and $\taa'$ are some inherited attributes. 
Assume that this upwards traversal yields that a proper ancestor $v'$ of $v$ is labeled by $g$.
To realize the tree translation in  Example~\ref{example 2},
$A$ must produce a symbol $d$ for every descendant of $v'$.
Now, the question is has $D$ done so for every descendant of $v'$?
If not then $D$ needs to determine for which descendants it still has to produce output (see Figure~\ref{fig-1}).
To do so $D$ needs to return to the node $v$. However $D$ cannot memorize the node $v$
from which it started and thus does not know how to return to $v$.

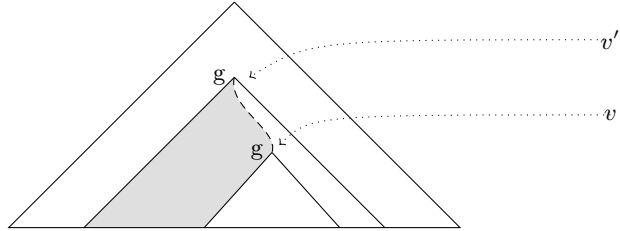
\begin{figure}[h!]
	\centering
\begin{tikzpicture}
	\fill[fill=gray!25]
	(0,-1) .. controls (-0.1,-1.4) and (0.6,-1.7) .. (0.5,-2)--(-0.4,-3)--(-2,-3)--(0,-1);
	
	\draw (0,0)--(-3,-3)--(3,-3)--(0,0);
	\draw (-0.2,-1) node {g};
	\draw (0,-1)--(2,-3);
	\draw (0,-1)--(-2,-3);
	\draw[densely dashed] (0,-1).. controls (-0.1,-1.4) and (0.6,-1.7) .. (0.5,-2);
	\draw (0.3,-2) node {g};
	\draw (0.5,-2)--(-0.4,-3);
	\draw (0.5,-2)--(1.4,-3);

	\draw[->,dotted] (5,-0.5) node{$v'$} .. controls (1,-0.5) .. (0.2,-1);
	\draw[->,dotted] (5,-1.5) node{$v$} .. controls (1,-1.5) .. (0.6,-1.9);

\end{tikzpicture}
\caption{
	Since $D$ visits nodes for which it has to produce an output symbol $d$ in reverse pre-order,
	$D$ must still produce a symbol
	$d$ for every node in the grayed section after
	determining that the ancestor $v'$ of $v$ is also labeled by $g$.}
\label{fig-1}
\end{figure}

Even giving $D$ access to look-ahead dos not remedy this problem.
Informally, this is because look-ahead can only provide
information about subtrees of $v$.
For the tree translation in Example~\ref{example 2} however, we require to know 
whether or not an ancestor of $v$ is labeled by~$g$. 

Observe that the knowledge
whether or not an ancestor of $v$ is labeled by $g$ in Example~\ref{example 2} can be acquired with regular look-around.

\begin{example}
	Let $U=(B,T)$ be a look-around where $B$ simply realizes the identity.
	Let $T'$ contain the rules
	\[
	\begin{array}{lclclcl}
		q(f (x_1, x_2)) & \rightarrow & f(q(x_1),q(x_2) ) &\quad & q'(f' (x_1, x_2))& \rightarrow & f'(q'(x_1),q'(x_2) )\\
		q(g (x_1, x_2))& \rightarrow  & g'(q'(x_1),q'(x_2) ) &\quad& 	q'(g' (x_1, x_2))& \rightarrow & g'(q'(x_1),q'(x_2) )\\
		q(e) & \rightarrow & e &\quad& q'(e') & \rightarrow & e'
	\end{array}
	\]
	Thus,
	all nodes $v$ of $s$ that are  either labeled by $g$ or are a descendant of such a node
	have a prime added to their label by $U$.
\end{example}

It should be clear that a  $datt^U$ $D$ using the relabeling $U$ in the previous example
can realize the tree translation in Example~\ref{example 2}.
In particular, to test whether or not an ancestor of a given node $v$ is labeled by $g$,
$D$ simply needs to test whether or not the parent node of $v$ is labeled by a symbol with a prime.

For completeness, note that the classes of tree translations realized by
functional $att$ and $datt^R$ are in fact incomparable as
shown in the following example.

\begin{example}\label{example 3}
	Consider the following tree translation: Let $L$ be a regular tree language over $\Sigma$.
	Let $\Delta$ consists of two symbols $'y'$ and $'n'$ that are both of rank $0$ and let
	$s\in T_\Sigma$.
	If $s\in L$ that we translate $s$ into $'y'$ otherwise $s$ is translated into $'n'$.
	Since  $L$ is a regular tree language, a deterministic bottom-up automaton $B$ accepting $L$
	exists. Using $B$ we can construct a  $datt$ with look-ahead that realizes
	the  translation above in a straight-forward manner. 
	
	Now consider a nondeterministic functional $att$ $A$. Before $A$ produces any output, it must check whether 
	its input tree $s$ is an element of $L$ or not. Specifically,  
	$A$ has to check whether $s\in L$ in a tree-walking-fashion, i.e.,
	$A$ tests $s$ like a \emph{tree-walking automaton} would. It is well known that
	tree-walking automata do not recognize all regular tree languages~\cite{DBLP:journals/siamcomp/BojanczykC08}.
	Thus, there are regular tree languages $L$ for which no functional $att$ realizing the above translation exists.
\end{example}

Denote by  $\mathcal{F}$
the class all functions. Furthermore, denote by $ATT^R$  and  $dATT^R$ the classes of tree translations realizable by  nondeterministic $atts^R$ and deterministic $atts^R$, respectively. Define $ATT$  and  $dATT$ analogously.
Examples~\ref{example 2} and~\ref{example 3} yield the following.

\begin{proposition}
	The following statements hold:
	\begin{enumerate}
		\item 	$dATT \subsetneq  ATT \cap \mathcal{F}$ and  $dATT^R \subsetneq  ATT^R \cap \mathcal{F}$.
		\item 	The classes $ATT \cap \mathcal{F}$ and $dATT^R$ are incomparable.
	\end{enumerate}
\end{proposition}

\section{Constructing a Uniformizer for an $\mathbf{Att}$}\label{construction}
Recall that before we can prove that for a given functional composition of $n$ $atts^U$
an equivalent composition of $n$ $datts^U$ can be constructed, 
we  prove that for a given $att^U$ $\hat{A}=(U,A)$,
a $datt^U$ $\hat{D}$ can be constructed such that $\tau_{\hat{D}}$ is a uniformizer of
$\tau_{\hat{A}}$. For the latter, we first show how to construct a $datt^U$ $D$ that realizes a uniformizer of $\tau_A$.

Subsequently, let $A=(S,I,\Sigma,\Delta,a_0, R)$ be  fixed.
Note that we allow the $att$~$A$ to be circular.

Before we construct the $datt^U$ $D$, we introduce the following definitions.
Let $s\in \text{dom} (A)$.
Let $t_1,\dots, t_n$ be trees in $T_\Delta [\text{SI} (s^\#)]$ and let
\[
\tau=(t_1 \Rightarrow_{A,s^\# } t_2 \Rightarrow_{A,s^\# }\cdots \Rightarrow_{A,s^\# } t_n).
\]
We call $\tau$ a  \emph{translation of $A$ on input $s$} if $t_1=a_0 (1)$ and $t_n\in T_\Delta$.
If  $\tau$ is a  translation of $A$ on input $s$, then the \emph{output of $\tau$} is $t_n$.
Furthermore, 
we say that $\tau$ contains a \emph{productive cycle} if $i<j\leq n$, a node $u$ and a proper descendant $u'$ of $u$
as well as $\ta (\nu) \in \text{SI} (s^\#)$ exist such that
$
t_i [u]=t_j[u']=\ta (\nu).
$
On the other hand, if  $i<\iota<j$  and a node $u$ exist such that
$
t_{i}[u]=t_{j}[u]=\ta (\nu) \quad\text{but}\quad t_{\iota}[u]\neq \ta (\nu),
$
then we say that  $\tau$ contains a \emph{non-productive cycle}.
We say that $\tau$ is \emph{cycle-free} if $\tau$ does not contain a cycle of either type.

\begin{observation}\label{observation 1}
	Let $s\in \text{dom} (A)$. Then a cycle free translation $\tau$ of $A$ on input $s$ exists.
\end{observation} 

Note that if $A$ is noncircular then any  translation of $A$ on input $s$ is cycle-free.
Let $\tau$ be a cycle-free translation of $A$ on input $s$.
Then clearly,
multiple instances of $\ta' (\nu') \in \text{SI} (s^\#)$ may occur
in $\tau$. This means that
since $A$ is nondeterministic, at distinct instances of $\ta' (\nu')$ distinct rules may be applied.
In the following, we say that $\tau$ is 
\emph{uniform} if at all such instances of $\ta' (\nu')$ the same rule is applied.
With Observation~\ref{observation 1}, it is easy to see that the following holds.

\begin{observation}\label{observation 2}
	Let $s\in \text{dom} (A)$. Then a cycle free translation $\tau$ of $A$ on input $s$ that is also uniform exists.
\end{observation}

\begin{example}\label{running example}
	Consider the following $att$ $A=(S,I,\Sigma,\Delta, a, R)$ where $S=\{a,a'\}$ and $I=\{b\}$. Let
	$\Sigma=\{ f^1, g ^1, h^1,e ^0\}$ and $\Delta= \{g^2, g'^2, h^1, e^0\}$.

	The rules for the root marker a given by
	$R_\#=\{ b (\pi 1) \rightarrow a' (\pi 1) \}$
	while the rules for the symbol $e$ are given by
	$R_e=\{ a (\pi) \rightarrow  b (\pi), a' (\pi) \rightarrow e \}$.
	Additionally, for the symbol $f$, we define
	\[
	\begin{array}{cc rll c rll c rll  c rll}
		R_f &=\{ & a (\pi) & \rightarrow & a (\pi 1) &\quad &a' (\pi 1) &\rightarrow & a (\pi 1) & \quad &a' (\pi) &\rightarrow & e &\quad & b (\pi 1) &\rightarrow & b (\pi) \ \}  \\
	\end{array}
	\]
	while for the symbol $g$, we define
	\[
	\begin{array}{cc rll c rll c rll  }
		R_g &=\{ & a (\pi) & \rightarrow & g (a (\pi 1),a(\pi 1)) &\quad &a' (\pi 1) &\rightarrow &  g (a (\pi 1),a(\pi 1)) & \quad &a' (\pi) &\rightarrow & e\\
		& & a (\pi) & \rightarrow & g' (a (\pi 1),a(\pi 1)) & \quad &  a' (\pi 1) &\rightarrow &  g' (a (\pi 1),a(\pi 1)) & \quad &b (\pi 1) &\rightarrow & b (\pi) \ \}.  \\
	\end{array}
	\]
	Finally, for the symbol $h$, we define
	\[
	\begin{array}{cc rll c rll c rll  c rll }
		R_h &=\{ & a (\pi) & \rightarrow & h (a (\pi 1)) &\quad &a' (\pi 1) &\rightarrow &  h (a (\pi 1)) & \quad &a' (\pi) &\rightarrow & e &\quad & b (\pi 1) &\rightarrow & b (\pi) \ \}. \\
	\end{array}
	\]
	Consider the trees $s_1= f(e)$, $s_2=h (e)$
	and $s_3=g (g (e))$.
	Consider the  translation 
	\[
	\begin{array}{lllllll lllll}
		a(1) & \Rightarrow_{A, s_1^\#} & a(1.1) &  \Rightarrow_{A, s_1^\#} & b(1.1) & \Rightarrow_{A, s_1^\#} & b(1) 
		&\Rightarrow_{A, s_1^\#} & a' (1) \\
		& \Rightarrow_{A, s_1^\#} & a(1.1) &  \Rightarrow_{A, s_1^\#} & b(1.1) 
		&\Rightarrow_{A, s_1^\#} & b(1)  & \Rightarrow_{A, s_1^\#} & a' (1)& \Rightarrow_{A, s_1^\#} & e
	\end{array}
	\]
	of $A$ on input $s_1$.
	Clearly this translation contains a non-productive cycle. Analogously, the translation 
	\[
	\begin{array}{lllllll}
		a(1) & \Rightarrow_{A, s_2^\#} & h (a(1.1))  & \Rightarrow_{A, s_2^\#}  & h(b(1.1)) & \Rightarrow_{A, s_2^\#} & h (b(1))\\
		&\Rightarrow_{A, s_2^\#} & h (a' (1)) &  \Rightarrow_{A, s_2^\#} & h (h (a(1.1))) & \Rightarrow_{A, s_2^\#}  &h (h (b(1.1))) \\
		&\Rightarrow_{A, s_2^\#} & h (h (b(1))) & \Rightarrow_{A, s_2^\#} & a' (1) & \Rightarrow_{A, s_2^\#} & h(h(e))
	\end{array}
	\]
	of $A$ on input $s_2$ contains a productive cycle.
	
	Tress that $s_3$ can be translated into include  $t_1= g (g(e,e), g'(e,e))$,  $t_2=  g' (g'(e,e), g'(e,e))$ and $t_3= g (g(e,e), g(e,e))$.
	It can be verified that the translations on input $s_3$ that output these trees are all cycle-free.
	Note however that the translation that outputs $t_1$ is however not uniform.
	In particular, two instances of $a(2)$ occur in this translation and while at one instance of $a(2)$ the rule $ a (\pi)  \rightarrow  g (a (\pi 1),a(\pi 1))$ is applied,
	the rule applied at the the other instance is $a (\pi)  \rightarrow  g' (a (\pi 1),a(\pi 1))$.
	The translations of $A$ on input $s_2$ that output~$t_2$ and~$t_3$ on the other hand are uniform.
\end{example}

Subsequently, denote by $\breve{\tau}_A$ the set of all pairs $(s,t) \in \tau_A$ for which
a uniform translation $\tau$ exists whose input is $s$ and whose output is $t$.
By definition of $\breve{\tau}_A$ and due to Observation~\ref{observation 2}, it should be clear that the following holds.

\begin{lemma}\label{uniform translation}
	$\breve{\tau}_A \subseteq \tau_A$ and 	$\text{dom} (\breve{\tau}_A) = \text{dom} (\tau_A)$. 
\end{lemma}

We remark that $\breve{\tau}_A$ is not necessarily a uniformizer of $\tau_A$ because $\breve{\tau}_A$
is not necessarily a function.

In the following,
we  construct a $datt^U$ $D$ that realizes a uniformizer of $\tau_A$ on the basis of
uniform translations.
More precisely, to construct $D$ we proceed as follows; 
First, we  show that given $A$,
a nondeterministic top-down relabeling $T$ and a $datt$ $D'$ can be constructed such that
the composition of $T$ and $D'$ simulates all uniform translations of $A$.
In particular, denote by $\tau_C$ the tree translation realized by the composition of 
$T$ and  $D'$, that is, 
\[
\tau_C= \{ (s,t) \mid (s,s')\in \tau_T \text{ and } (s',t) \in \tau_{D'} \}.
\]
We can construct  $T$ and $D'$ such that $\tau_C = \breve{\tau}_A$.
Given $T$ and $D'$, we show how to construct a $datt^U$ $D$ such that $D$ realizes a uniformizer  
of $\tau_C$. Due to Lemma~\ref{uniform translation}, it follows that the $datt^U$ $D$ also realizes a uniformizer  
of $\tau_A$.

Before constructing $T$ as well as $D'$ we introduce the following definition.
Consider the set $R_\sigma$, i.e., the set of all rules of $A$ for the symbol $\sigma$.
Let $\bar{R} \subseteq R_\sigma$. We call the set $\bar{R}$ \emph{unambiguous} if 
$\bar{R}$ does not contain two distinct rules with the same left-hand side.\\ 

\noindent
\textbf{Idea of $\mathbf{T}$ and $\mathbf{D'}$.}\quad
In the following we describe our idea: 
To simulate uniform translations of $A$, the nondeterministic top-down relabeling $T$ determines for every input node $v$ which rules 
will be used at the node $v$. The $datt$ $D$ then only has to apply these rules.

More precisely, denote by $\tilde{\Sigma}$ the output alphabet of $T$ and the input alphabet of $D'$.
The alphabet $\tilde{\Sigma}$ consists of symbols of the form $\langle \sigma,\bar{R} \rangle$
where $\sigma\in \Sigma$ and $\bar{R}$ is an unambiguous subset of $R_\sigma$.
The symbol  $\langle \sigma,\bar{R} \rangle$ is of rank $k$ if $\sigma$ is.
By relabeling a node by $\langle \sigma,\bar{R} \rangle$, $T$ signals to $D'$ that only rules in $\bar{R}$
are allowed to be used at that node.
Note that since $\bar{R}$ is an unambiguous subset of $R_\sigma$, it follows easily that if
$(s,\tilde{s}) \in \tau_T$ and $\tau'$ is a translation of $D'$ on input $\tilde{s}$,
then $\tau'$ simulates a translation $\tau$ of $A$ on input $s$ that is uniform.\\

\noindent
\textbf{Construction of $\mathbf{T}$.}\quad
Subsequently, we define the nondeterministic top-down relabeling 
$T=(\{q\},\emptyset,\Sigma, \tilde{\Sigma}, q, \tilde{R})$.
To do so, all we need to do is to define the set $\tilde{R}$.
Specifically, we define that
\[
q (\sigma (x_1,\dots,x_k) )\rightarrow \langle \sigma, \bar{R} \rangle (q (x_1),\dots,q (x_k)) \in \tilde{R}
\]
for all $\sigma\in \Sigma_k$, $k\geq 0$ and for all unambiguous $\bar{R} \subseteq R_\sigma$.\\

\noindent
\textbf{Construction of  $\mathbf{D'}$.}\quad
	The $datt$ $D'$ is constructed in a straight-forward manner.
	In particular, we define $D'=(S,I,\tilde{\Sigma},\Delta,a_0, R')$.
	
	Recall that due to Proposition~\ref{assumption}, we can assume that $R_\#$ is deterministic, i.e., the set of rules of
	$A$ for the root marker,  contains no distinct rules with the same left-hand side. In other words, we can assume that
	$R_\#$ is unambiguous. Therefore, we define  $R'_\# = R_\#$.
	
	Recall that $\tilde{\Sigma}$ consists of symbols of the form $\langle \sigma,\bar{R} \rangle$
	where $\sigma \in \Sigma$ and    $\bar{R}$ is an unambiguous subset of $R_\sigma$.
	For the symbol $\langle \sigma,\bar{R} \rangle$ we define
	$R'_{\langle \sigma,\bar{R} \rangle} = \bar{R}$.
	Since by definition of the alphabet $\bar{\Sigma}$, the set $\bar{R}$ is unambiguous it should be clear that $D'$
	is deterministic.

\begin{example}\label{running example 2}
	Consider the $att$ $A$ in Example~\ref{running example}.
	Consider the symbol $f$ and the set of rules $R_f$ of $A$. Then the unambiguous subsets of $R_f$ are
	\[
	\begin{array}{l c l}
		 \{a (\pi) \rightarrow a (\pi 1),\  b (\pi 1) \rightarrow b (\pi),\ a' (\pi 1) \rightarrow  a (\pi 1)\} & \quad &	\{b (\pi 1) \rightarrow b (\pi),\ a' (\pi 1) \rightarrow  e\}\\
		\{a (\pi) \rightarrow a (\pi 1),\  b (\pi 1) \rightarrow b (\pi),\ a' (\pi 1) \rightarrow  e\} & \quad & \{a (\pi) \rightarrow a (\pi 1)\}\\
		 \{a (\pi) \rightarrow a (\pi 1),\ a' (\pi 1) \rightarrow  a (\pi 1)\} & \quad & \{a' (\pi) \rightarrow  e\}\\
		 \{a (\pi) \rightarrow a (\pi 1),\ a' (\pi 1) \rightarrow  e\} & \quad & \{ a' (\pi 1) \rightarrow  a (\pi 1)\}\\
		 \{a (\pi) \rightarrow a (\pi 1),\ b (\pi 1) \rightarrow b (\pi)\} & \quad &  \{b (\pi 1) \rightarrow b (\pi)\}\\
		\{b (\pi 1) \rightarrow b (\pi),\ a' (\pi 1) \rightarrow  a (\pi 1)\} & \quad &\emptyset. 
	\end{array}
	\]
	To construct the nondeterministic top-down relabeling $T=(\{q\},\emptyset,\Sigma, \tilde{\Sigma}, q, \tilde{R})$ from  $A$,  we define for each of these  subsets $\bar{R}$ an output symbol of the form  
	$\langle f,\bar{R} \rangle$ and a
	rule of the
	form $q (f (x_1)) \rightarrow \langle f, \bar{R} \rangle (q (x_1))$.
	For instance for the sets $\{a (\pi) \rightarrow a (\pi 1),\ b (\pi 1) \rightarrow b (\pi)\}$ and $\{a (\pi) \rightarrow a (\pi 1)\}$,
	we define the symbols
	\[
	\langle f, \{a (\pi) \rightarrow a (\pi 1),\ b (\pi 1) \rightarrow b (\pi)\} \rangle\quad \text{and}\quad
	\langle f, \{a (\pi) \rightarrow a (\pi 1)\}\rangle
	\]
	along with the rules
	\[
	q (f (x_1)) \rightarrow \langle f, \{a (\pi) \rightarrow a (\pi 1),\ b (\pi 1) \rightarrow b (\pi)\} \rangle (q (x_1))
	\]
	and 
	\[
	q (f (x_1)) \rightarrow \langle f, \{a (\pi) \rightarrow a (\pi 1)\}\rangle (q (x_1)),
	\] 
	respectively. For the rule sets $R_g$, $R_h$ and $R_e$ of $A$, we proceed analogously.
	
	To construct of the $datt$ $D'=(S,I,\tilde{\Sigma},\Delta,a_0, R')$ from  $A$, we define
	 for the output symbols $\langle f, \{a (\pi) \rightarrow a (\pi 1),\ b (\pi 1) \rightarrow b (\pi)\} \rangle$
	and $\langle f, \{a (\pi) \rightarrow a (\pi 1)\}\rangle$ of~$T$ for example, the rule sets
	\[R'_{\langle f, \{a (\pi) \rightarrow a (\pi 1),\ b (\pi 1) \rightarrow b (\pi)\} \rangle}=\{a (\pi) \rightarrow a (\pi 1),\ b (\pi 1) \rightarrow b (\pi)\}\]
	and 
	\[R'_{\langle f, \{a (\pi) \rightarrow a (\pi 1)\} \rangle}=\{a (\pi) \rightarrow a (\pi 1)\},\]
	respectively. We proceed analogously for the remaining output symbols.
	Note that $R'_{\#}=R_\#$.
\end{example}

Recall that $\tau_C$ denotes the tree translation realized by the composition of 
$T$ and  $D'$ and that $\breve{\tau}_A$ is the set of all pairs $(s,t) \in \tau_A$ for which there is  a
a uniform translation. We now prove the following statement.

\begin{lemma}\label{composition correct}
	The sets $\breve{\tau}_A$ and $\tau_C$ are equal.
\end{lemma}

\begin{proof}
To begin our proof we show that $\tau_C \subseteq \breve{\tau}_A$.
Let $(s, \tilde{s}) \in \tau_{T}$ and let $(\tilde{s}, t) \in \tau_{D'}$.
Let 
\[
a_0 (1) = t_1 \Rightarrow_{D', \tilde{s}^\#} t_2 \Rightarrow_{D', \tilde{s}^\#} \cdots \Rightarrow_{D', \tilde{s}^\#} t_n=t 
\]
be the corresponding translation of $D'$ on input $\tilde{s}$. 
Let $\langle\sigma, \bar{R} \rangle$ be an arbitrary symbol in $\tilde{\Sigma}$.
Recall that this means that
$\sigma\in \Sigma$ and $\bar{R}$ is an unambiguous subset of $R_\sigma$.
By construction of
$D'$, $R'_{\langle\sigma, \bar{R} \rangle} \subseteq R_{\sigma}$.
Hence,  a translation $\tau$  of $A$ on input $s$ exists that is of the form
\[
\tau= (a_0 (1) = t_1 \Rightarrow_{A, s^\#} t_2 \Rightarrow_{A, s^\#} \cdots \Rightarrow_{A, s^\#} t_n=t) 
\]
Since $D'$ is deterministic, $\tau$ is a uniform translation.

To show the converse, that is, to show that $\tau_C \supseteq \breve{\tau}_A$, consider a uniform
translation
\[
\tau= (a_0 (1) = t_1 \Rightarrow_{A, s^\#} t_2 \Rightarrow_{A, s^\#} \cdots \Rightarrow_{A, s^\#} t_n=t) 
\]
of $A$ on input $s\in T_\Sigma$.
Let $v\in V(s)$ and let $s[v]= \sigma \in \Sigma_k$, $k\geq 0$. Let $a\in S$ and $a(\pi)\rightarrow \xi \in R_\Sigma$.
Recall that for a right-hand side $\xi$ of a rule in $R_\sigma$ and a node $v$,
$\xi [\pi\leftarrow v]$ denotes the tree obtained by replacing all occurrences
of $\pi$ by $v$ (cf. Section~\ref{preliminaries}).

In the following, we say that the rule $a(\pi)\rightarrow \xi$ is \emph{used in $\tau$ at the input node $v$} if
$i\in [n]$ and a node $u\in V(t_i)$ exist such that $t_i[u] = a(1.v)$ and 
$t_{i+1}=t_i[u \leftarrow \xi[\pi \leftarrow 1.v ] ]$. 

Analogously, we say that the rule $b (\pi j)\rightarrow \xi$, where $b\in I$ and $k\in [k]$, 
is \emph{used in $\tau$ at the input node $v$}  if
$i\in [n]$ and a node $u\in V(t_i)$ exist such that $t_i[u] = b(1.v.j)$ and 
$t_{i+1}=t_i[u \leftarrow \xi[\pi \leftarrow 1.v ] ]$.

Note that the reason why we have $a(1.v)$ and $ b(1.v.j)$ in the definitions above is
because translations of $A$ always involve the root marker.

Denote by $\tau[v]$ the set of all rules used in $\tau$ at the input node $v$.
Note that since $\tau$ is a uniform translation, $\tau[v]$ is an unambiguous subset of $R_\sigma$.
By construction of $T$, it should be clear that $T$ can transform the tree $s$ into the tree $\tilde{s}$
over $\tilde{\Sigma}$ such that if the node $v$ is labeled by the symbol $\sigma$ in $s$
then $v$ is labeled by $\langle\sigma,\tau[v] \rangle$ in $\tilde{s}$.
By definition of $D'$ it follows that there is a translation of $D'$ on input $\tilde{s}$
that outputs $t$. Hence, our lemma follows.
\end{proof}

Before we construct the $datt^U$ $D=(U,D')$ from $T$ and $D'$, note that we 
can assume that $\text{range} (T) \subseteq \text{dom} (D')$.
This is because by
Proposition~\ref{domain} (see also~\cite{DBLP:journals/ipl/FulopM00}), 
the domain of the $datt$ $D'$ is effectively regular.
Hence a (nondeterministic) top-down automaton $M$ recognizing $D'$ exists.
In order to restrict the range of $T$ to the domain of $D'$, we simply run $M$ in parallel with
$T$ in the usual way, that is, we construct a new top-down relabeling 
from $T$ and $M$ using the product construction (cf~\cite[Definition on p.195]{DBLP:journals/iandc/Baker78b}).\\

Furthermore, note by~\cite{DBLP:journals/ipl/Engelfriet78},
that the following holds  for top-down tree transducers, i.e., $att$ without inherited attributes.

\begin{proposition}\label{engelfriet}
	For any top-down tree transducer $M$, a deterministic top-down tree transducer $M'$ with regular-look-ahead
	can be constructed such that $M'$ realizes a uniformizer of $\tau_M$. 
\end{proposition}

\begin{proof}
	Subsequently we briefly sketch the idea of the procedure given in~\cite{DBLP:journals/ipl/Engelfriet78}.
	Consider the top-down tree transducer $M$.
	 Denote by $\Sigma$ and $\Delta$ the input and output alphabet of $M$, respectively.
	Let $R$ be the set of rules of $M$
	Assume a total order on the rules in $R$.
	
	Given $M$, we construct $M'$ such that $M'$ has the same states as $M$.
	With theses states $M'$ simulates $M$ as follows:
	Let the state $q$ process a node labeled by $\sigma\in \Sigma$.
	Let $\rho_1,\dots,\rho_n$ be the rules of $M$ for the symbol $\sigma$ where $q$ occurs on the left-hand side.
	To simulate $M$, $M'$ checks which of the rules $\rho_1,\dots,\rho_n$ is \emph{applicable}, that is, which of these
	rules lead to the generation of a ground output tree (i.e., a tree in $T_\Delta$).
	Whether or not a rule is applicable can be tested by $M'$ using its look-ahead.
	Among the applicable rules, $M'$ then picks the first one according to the given total order on $R$
	to simulate $M$.
\end{proof}

Note that if $M$ is a top-down relabeling, then the top-down tree transducer $M'$ with regular-look-ahead that
the procedure in~\cite{DBLP:journals/ipl/Engelfriet78} yields is in fact a
top-down relabeling with regular-look-ahead, i.e., a look-around.\\

\noindent
\textbf{Construction of the $\mathbf{datt^U \ D=(U,D')}$.}\quad
To construct the $datt^U$ $D$, all we need to do is to construct $U$ from $T$.
Note that by definition,
$T$ is also a  top-down tree transducer.
Therefore, due to Proposition~\ref{engelfriet},
let  $U$ be constructed as one of the top-down tree transducers  with regular-look-ahead
that realizes a uniformizer of $\tau_T$.
Since $T$ is a top-down relabeling, previous observations yield that
$U$ is a look-around. 
This concludes the construction of the $datt^U$ $D$.

\begin{example}
	Consider the nondeterministic top-down relabeling $T$ and the $datt$ $D'$ in Example~\ref{running example 2}.
	It can be verified that the domain of $D'$ consists of trees $s$ for which the following statements hold:
	\begin{enumerate}
		\item Let $v$ be the root of $s$. Then $v$ is  either labeled by 
		\begin{enumerate}
			\item $\langle e, R_e \rangle$ or 
			\item by a symbol of the form $\langle \sigma, \bar{R} \rangle$, where $\sigma\in \{f,g,h\}$. The set
			$\bar{R}$ is a subset of $R_\sigma$ that contains the rule $a'\rightarrow e$ as well as a rule  in which $a (\pi)$ occurs on the left-hand side
			and a rule  in which $b(\pi 1)$ occurs on the left-hand side.
		\end{enumerate}
		\item All other nodes $v'$  of $s$, are  either labeled by 
		\begin{enumerate}
			\item $\langle e, \bar{R}' \rangle$ where $\bar{R}$ is a subset of $R_e$  that contains the rule $a(\pi)\rightarrow \b (\pi)$  or 
			\item by a symbol of the form $\langle \sigma, \bar{R} \rangle$, where $\sigma\in \{f,g,h\}$ and the set
			$\bar{R}$ is a subset of $R_\sigma$
			that contains a rule  in which $a (\pi)$ occurs on the left-hand side
			and a rule  in which $b(\pi 1)$ occurs on the left-hand side.
		\end{enumerate}
	\end{enumerate}
	It is easy to see that the rules of $T$ can be modified such that $T$  only outputs  trees for which
	the previous statements hold.
	
	Recall that by~\cite{DBLP:journals/ipl/Engelfriet78}, Proposition~\ref{engelfriet} holds.
	Hence,  top-down tree transducer  with regular-look-ahead
	that realizes a uniformizer of $\tau_T$.
	In fact one of the possible  top-down tree transducer  with regular-look-ahead that the procedure of~\cite{DBLP:journals/ipl/Engelfriet78}
	yields is $U=(B,T')$ where $B$ is realizes the identity and $T'=(\{q_0,q\}, \emptyset,\Sigma, \tilde{\Sigma}, q_0,\tilde{R}')$.
	The set $\tilde{R}'$ contains the rules
	\[
	\begin{array}{cll c cll}
		q_0 (f (x_1)) & \rightarrow & \langle f, \bar{R}_f \rangle ( q (x_1)) &\quad&
		q_0 (g (x_1)) & \rightarrow & \langle g, \bar{R}_g \rangle ( q (x_1)) \\
		q_0 (h (x_1)) & \rightarrow & \langle h, \bar{R}_h \rangle ( q (x_1)) &\quad&
		q_0 (e) & \rightarrow & \langle e, \bar{R}_e \rangle ( q (x_1)) \\
	\end{array}
	\]
	along with the rules
	\[
	\begin{array}{cll c cll}
		q (f (x_1)) & \rightarrow & \langle f, \bar{R}_f \rangle ( q (x_1)) &\quad&
		q (g (x_1)) & \rightarrow & \langle g, \bar{R}_g \rangle ( q (x_1)) \\
		q (h (x_1)) & \rightarrow & \langle h, \bar{R}_h \rangle ( q (x_1)) &\quad&
		q (e) & \rightarrow & \langle e, \bar{R}_e \rangle ( q (x_1)) \\
	\end{array}
	\]
	where
	\begin{itemize}
		\item  $\bar{R}_f= \{a_{\pi} (\pi)\rightarrow e, a(\pi) \rightarrow a(\pi 1),\ b(\pi1)\rightarrow  b(\pi)\}$,
		\item $\bar{R}_g=\{ a_{\pi} (\pi)\rightarrow e, a_{\pi} (\pi)\rightarrow e,\ a(\pi) \rightarrow g' (a(\pi 1), a(\pi 1)),\ b(\pi1)\rightarrow  b(\pi)\}$,
		\item $\bar{R}_h=  \{a_{\pi} (\pi)\rightarrow e,\ a(\pi) \rightarrow h (a(\pi 1)),\ b(\pi1)\rightarrow  b(\pi)\}$, and
		\item $\bar{R}_e=\{a_{\pi} (\pi)\rightarrow e, a(\pi) \rightarrow  b(\pi)\}$.
	\end{itemize}
	Note that $U$ is a look-around. Hence, together with $D'$, $U$ yields a $datt^U$ that realizes a unifomizer of $\tau_A$.
\end{example}

\begin{lemma}\label{uniformizer att}
	The $datt^U$ $D=(U,D')$ realizes a uniformizer of $\tau_A$.
\end{lemma}

\begin{proof}
	Due to Lemmas~\ref{uniform translation} and~\ref{composition correct}, it is obviously sufficient to show that
	$D$ realizes a uniformizer of $\tau_C$. To start with, note that obviously $\tau_D$ is a function since $D$ is deterministic.
	
	By construction of $D$, it follows obviously that $\tau_D \subseteq \tau_C$.
	In particular, this follows since by construction $U$ realizes a uniformizer of $\tau_T$ and hence
	$\tau_U\subseteq \tau_T$.
	
	We now show that $\text{dom} (\tau_C) = \text{dom} (D)$.
	First note that $\tau_D \subseteq \tau_C$ implies $\text{dom} (\tau_C) \subseteq \text{dom} (D)$.
	For the converse, let $s \in \text{dom} (\tau_C)$.
	Obviously, this means that $s\in \text{dom} (T)$. Since $U$ realizes a uniformizer of $\tau_T$, 
	it has  the same domain as $\tau_T$.
	Therefore, it follows that $s\in \text{dom} (U)$. Let  $(s,\tilde{s}) \in \tau_U$.
	Recall our previous assumption that $\text{range} (T) \subseteq \text{dom} (D')$.
	Obviously, this means that $\text{range} (U) \subseteq \text{dom} (D')$ and hence $\tilde{s} \in \text{dom} (D')$.
	This in turn yields that  $s \in \text{dom} (D)$ and hence $\text{dom} (\tau_C) = \text{dom} (D)$.
\end{proof}

\noindent
Due to Lemma~\ref{uniformizer att}, the following holds.

\begin{theorem}\label{functional}
	For any  $att$ $A$ a $datt^U$ $D$ can be constructed such that $D$ realizes a uniformizer of $\tau_A$.
\end{theorem}

\section{From Unformizers for $\mathbf{Att}$ to Uniformizers for $\mathbf{Att^U}$}\label{attu section}
Consider an arbitrary $att^U$ $\hat{A}=(U,A)$.
In the following, we show how to construct a $datt^U$ $\hat{D}$ such that $\tau_{\hat{D}}$ is a uniformizer of
$\tau_{\hat{A}}$.
Consider the underlying $att$ $A$ of $\hat{A}$.
Due to Theorem~\ref{functional}, a $datt^U$ $\hat{D}'=( U', D)$ that realizes a uniformizer of $\tau_A$ can be constructed.
Clearly
\[
\begin{array}{cl}
	&\{(s,t) \mid (s,\breve{s}) \in \tau_U \text{ and } (\breve{s},t)\in \tau_{\hat{D}'}\}\\
	=&\{(s,t) \mid (s,\breve{s}) \in \tau_U \text{ and } (\breve{s},\hat{s})\in \tau_{U'} \text{ and }
	(\hat{s},t) \in \tau_{D}\}
\end{array}
\]
is a uniformizer of $\tau_{\hat{A}}$. Consider the look-arounds $U$ and $U'$.
Due to Theorem~2.6 of~\cite{DBLP:journals/mst/Engelfriet77} and Theorem~1 of~\cite{DBLP:journals/iandc/Baker78b} (and their proofs),
look-arounds  are closed under composition. In other words, a look-around $\hat{U}$ can be constructed such that 
\[
\tau_{\hat{U}}=\{ (s,\hat{s}) \mid  
(s,\breve{s}) \in \tau_U \text{ and }(\breve{s},\hat{s})\in \tau_{U'}\}.
\]
Therefore,
the  $datt^U$ $(\hat{U}, D)$ realizes a uniformizer of $\tau_{\hat{A}}$. This yields the following.

\begin{theorem}\label{uniformizer}
	For any $att^U$ $\hat{A}$ a $datt^U$ that realizes a uniformizer of $\tau_{\hat{A}}$ can be constructed.
\end{theorem}

\section{Final Results}
Analogously as in~\cite{DBLP:journals/ipl/Engelfriet78} for top-down tree transducers,
we obtain the following result using uniformizers.

\begin{theorem}\label{final result}
Let $C$ be a composition of $n$ $atts^U$. Then $datts^U$ $D_1,\dots, D_n$ can be constructed such that
$C$ and the composition of $D_1,\dots, D_n$ are equivalent.
\end{theorem}

\begin{proof}
	Consider the $atts^U$ $A_1,\dots,A_n$.
	Let the composition $C$ of $A_1,\dots,A_n$ be functional, i.e., let
	\[
	\tau_C=\{ (s,t) \mid (s,t_1) \in \tau_{A_1}, (t_{i-1},t_{1})\in \tau_{A_i} \text{ for } 1<i<n \text{ and }(t_{n-1},t) \in \tau_{A_n} \}
	\]
	be a function.
	For $i<n$, we can assume that  $\text{range} (A_i)\subseteq \text{dom} (A_{i+1})$ holds.
	In particular, by 
	Proposition~\ref{domain},  $\text{dom} (A_{i+1})$ is effectively recognizable, i.e., a (nondeterministic) top-down automaton
	$M_{A_{i+1}}$ recognizing $\text{dom} (A_{i+1})$ can be constructed. 
	By running $M_{A_{i+1}}$ in parallel to
	$A_i$, it should be clear that we can restrict the range of $A_i$ to trees in $\text{dom} (A_{i+1})$.
	
	By Theorem~\ref{uniformizer}, let $D_1,\dots,D_n$ be  $datts^U$ realizing uniformizers of $\tau_{A_1},\dots,\tau_{A_n}$, respectively.  
	Denote by $C_d$ the composition of  $D_1,\dots,D_n$. Obviously $\tau_{C_d} \subseteq \tau_C$.
	Note that  $\text{range} (A_i)\subseteq \text{dom} (A_{i+1})$ implies  $\text{range} (D_i)\subseteq \text{dom} (D_{i+1})$
	for $i<n$. 
	In particular, since $D_i$ and $D_{i+1}$ realize uniformizers of $A_i$ and $A_{i+1}$, respectively,
	\[
	\text{range} (D_i)\subseteq \text{range} (A_i)\subseteq \text{dom} (A_{i+1}) =  \text{dom} (D_{i+1}).
	\]
	Note that $\text{dom} (D_1) =\text{dom} (A_1)$.
	Thus the fact that $\text{range} (D_i)\subseteq \text{dom} (D_{i+1})$ for $i<n$, yields that $C_d$ and $C$ have the same domain.
	
	Since $\tau_C$ is a function and $\tau_{C_d} \subseteq \tau_C$, the latter yields that $\tau_{C_d}= \tau_C$
	which in turn yields Theorem~\ref{final result}.
\end{proof}

Denote by $(ATT^U)^n$  and  $(dATT^U)^n$ the classes of tree translations realizable by the composition of
$n$ nondeterministic $atts^U$ and $n$ deterministic $atts^U$, respectively. Theorem~\ref{final result} yields the following.

\begin{theorem}
	$(ATT^U)^n \cap \mathcal{F} = (dATT^U)^n$.
\end{theorem}

Note that if an $att$ $A$ (with or without look-around) is functional, then any $datt^U$ that realizes a uniformizer of $\tau_A$
is in fact equivalent to $A$. Hence, Theorems~\ref{functional} and~\ref{uniformizer} yield the following

\begin{corollary}
	For any functional $att$ (with or without look-around) an equivalent $datt^U$ can be constructed.
\end{corollary}

\section{Conclusion}
Consider an arbitrary composition $C$  of $n$ attributed tree transducer with look-around
that realizes a function.
In this paper we have provided a procedure which given $C$, computes 
$n$ deterministic attributed tree transducer with look-around $D_1,\dots, D_n$
such that $C$ and the composition of $D_1,\dots, D_n$ are equivalent.
To do so we have shown that  any attributed tree transducer $A$ admits a uniformizer realized a deterministic 
attributed tree transducer with look-around  can be constructed.
An obvious question is: Do we always need look-around? 
One wonders when a uniformizer of $A$ can be implemented by a deterministic attributed tree transducer
\emph{without} look-around? This question is addressed in~\cite{DBLP:conf/icalp/FiliotJLW16} for finite-valued string transducers.
Specifically,
for such transducers it is decidable whether or not a uniformizer realized by a deterministic string transducer exists.

Given our result another question 
is: 
Given a composition of attributed tree transducer $C$
is it decidable whether or not $C$ realizes a function? To the best of our knowledge this is an open problem.
Note that whether or not a composition of top-down tree transducers is functional has recently been
shown to be decidable~\cite{DBLP:journals/iandc/ManethSV24}.
In contrast,
even for a single attributed tree transducer it is unknown whether or not functionality is decidable.
Note that decidability of the latter would imply that equivalence of deterministic attributed tree transducers is decidable.
The latter is a long standing open problem.

\bibliographystyle{splncs04}
\bibliography{mybib}

\end{document}